\documentclass[letter,pra,twocolumn,superscriptaddress,nofootinbib]{revtex4-1}
\usepackage{amsmath, amsthm, amssymb}
\usepackage{latexsym}
\usepackage[utf8]{inputenc}
\usepackage{color}
\definecolor{earth}{rgb}{0.549,0,0}

\newtheoremstyle{plain}
  {}   
  {}   
  {\itshape}  
  {}       
  {\itshape} 
  {.}         
  {5pt plus 1pt minus 1pt} 
  {}          

\newtheoremstyle{definition}
  {}   
  {}   
  {\normalfont}  
  {}       
  {\itshape} 
  {.}         
  {5pt plus 1pt minus 1pt} 
  {}          

\renewcommand\bm[1]{#1}
\usepackage{graphics,epstopdf,subfigure,color,hyperref}
\definecolor{Red}{rgb}{1,0,0}

\def\bra#1{\mathinner{\langle{#1}|}}
\def\ket#1{\mathinner{|{#1}\rangle}}

\def\ketbra#1#2{{\ket{#1}}{\bra{#2}}}

{\catcode`\|=\active
  \gdef\Braket#1{\begingroup
\mathcode`\|32768\let|\BraVert\left<{#1}\right>\endgroup}}
\def\BraVert{\egroup\,\mid\,\bgroup}

\def\sh{\mathcal{H}}

\newtheoremstyle{plain}
  {}   
  {}   
  {\itshape}  
  {}       
  {\itshape} 
  {.}         
  {5pt plus 1pt minus 1pt} 
  {\thmname{#1}\thmnumber{ #2}\thmnote{. #3}}          

\newcounter{commco}

\newtheorem{proposition}{Proposition}
\newtheorem{theorem}[proposition]{Theorem}
\newtheorem{lemma}[proposition]{Lemma}

\theoremstyle{definition}
\newtheorem{definition}[proposition]{Definition}
\theoremstyle{remark}
\newtheorem{remark}[proposition]{Remark}

\def\ncms{{n_{\openone}}}

\begin{document}
\title{Entanglement and deterministic quantum computing with one qubit}

\author{Michel Boyer}
\email{boyer@iro.umontreal.ca}
\affiliation{DIRO, Universit\'e de Montr\'eal, Canada}
\author{Aharon Brodutch}
\email{brodutch@physics.utoronto.ca}
\affiliation{Institute for Quantum Computing, University of Waterloo, Waterloo N2L 3G1, Ontario, Canada}
\affiliation{Department of Physics and Astronomy, University of Waterloo, Waterloo, Ontario N2L 3G1, Canada}
\author{Tal Mor}
\email{talmo@cs.technion.ac.il}
\affiliation{Computer Science Department, Technion, Israel}

\date{\today}

\begin{abstract}
 
The role of entanglement and quantum correlations in complex physical systems and quantum information processing devices has become a topic of intense study in the past two decades. In this work we present new tools for learning about entanglement and quantum correlations in dynamical  systems where the quantum states are mixed and the eigenvalue spectrum is highly degenerate. We apply these results to the \emph{Deterministic quantum computing with one qubit} (DQC1)  computation model and show that the states generated in a DQC1 circuit  have an eigenvalue structure that makes them difficult to entangle, even when they are relatively far from the completely mixed state.  Our results strengthen the conjecture that it may be possible to find quantum algorithms that do not generate entanglement and yet still have an exponential advantage over their classical counterparts. 
\end{abstract}
\maketitle

\section{Introduction}

Quantum computing is a promising theoretical field that mainly deals with problems that are hard for regular computers
 and are easy for a quantum computer \cite{NielsenandChuang, Steane1998,  Bernstein1997}. Experiments in quantum computing devices did not yet employ many qubits, and the most successful approaches can only deal with 7-14 qubits \cite{SCerror,Negrevergne2006,Lu,Monz2011}. While many  experimental and theoretical efforts are geared towards the implementation of universal quantum machines with more qubits, there could  be  considerable advantages with other, less powerful models which we call \emph{Semi quantum computer} (SQC) models. 
 
SQCs   outperform classical computers for some tasks while being less powerful and potentially less  technologically demanding than universal quantum  computers.   For more than a decade, there have been various descriptions of  SQCs  that could demonstrate a quantum advantage (sometimes called \emph{quantum supremacy} \cite{Preskill2012} ). In some cases these models don't even require or produce entanglement \cite{arXiv:quant-ph/030618,Alastair15,arXiv:0906.3656,KMR05}.

Among the candidate SQC models are those such as  quantum annealers  \cite{Boixo2014},  that may be able to beat classical computers in a benchmarking experiment,  but do not seem to offer an exponential advantage (in the input size). Other candidate models, which we call Sub-universal quantum computers (SuQCs) probably do offer an exponential advantage for a few specific tasks but remain sub universal. Three   prominent examples of  SuQCs are the linear optics model   \cite{Aaronson2010}, instantaneous quantum computing \cite{Shepherd2008}, and deterministic quantum computing with 1 qubit (DQC1)\cite{KL}.

For each (candidate) SuQC model there are three important questions we can ask: 
\begin{enumerate}
\item What hard problems can it solve efficiently  and are these problems interesting (e.g for science or industry) ? 
\item Is it indeed much more feasible experimentally, and if so, why?
\item What are the conceptual properties that make it (exponentially) more powerful than a classical computer.
\end{enumerate}

Briefly let us mention the status  of 
DQC1 
in connection to the first two questions:
A  generic problem for the DQC1 model is  the trace estimation problem, i.e  estimating the normalized trace of an efficiently implementable unitary operation.  For some subsets of $n$-qubit unitaries this problem  is believed to be hard for a classical computer \cite{Datta2005,Shor2007}, i.e there is good evidence that it impossible to find  a classical polynomial time (in $n$) algorithm for estimating the normalized trace.  For example, it is possible to use the DQC1 model to  estimate  Jones polynomial at a fifth root of unity for the trace closure of a braid.  This special case is  complete for DQC1 and   was  implemented experimentally on 4 qubits \cite{Passante2009}.  The model is, however, and idealization as it assumes no errors in the implementation. It is currently unknown if DQC1 is a viable model that can be scaled up in the presence of imperfections.

Here we deal with the third question with regards to DQC1.  Our approach is to study the space of quantum states that can be generated during the computational process, and in particular study the correlations within these  states. This approach follows earlier works on the circuit model \cite{Jozsa2003,arXiv:quant-ph/0301063}, DQC1 \cite{Datta2007,DattaShajiCaves08,datta,arXiv:1109.5549,arXiv:quant-ph/0110029,arXiv:1006.2460},  measurement based quantum computing \cite{Rieffel2014},  and other models  \cite{arXiv:0906.3656}.  As in these works, we do not take error correction into account.

  For pure state quantum circuits,  
Vidal \cite{arXiv:quant-ph/0301063} showed that a circuit can be efficiently simulated on a classical computer when  the amount of entanglement  
is not  large (i.e the maximal Schmidt rank over all bipartition grows slower than $\log(n)$) at all times. 
Josza and Linden \cite{Jozsa2003} showed that there is an efficient  classical simulation of a quantum 
circuit  if  the register has a $p$ block structure for constant $p$ at all times. An n-qubit state $\rho$ has 
a $p$ block structure if $\rho=\bigotimes_j\varrho_j$ and $\varrho_j$ are  $m_j \le p$ qubit states. This 
structure implies that the blocks are not correlated.  When  $\rho$ is pure,  the lack of correlations 
means that there is no entanglement  between  blocks of qubits.  Preliminary results on the DQC1   
model \cite{DattaShajiCaves08,DiscordRMP} have given some reason to suspect that discord, a more general 
measure of quantum correlations, may play  a similar role to pure state entanglement 
in this context. 

For any measure of  quantum correlations (e.g  entanglement  or  discord) and any specific way of quantifying these correlations  $\mathcal{C}$,  we may ask if one of the following is a necessary condition for computational speedup: 

{\it
\begin{enumerate}
\item At some point in the algorithm the register must have large amounts of quantum correlations $\mathcal{C}$ for at least one bipartition. \label{S1} 
\item At some point in the algorithm the register must have some quantum correlations for exponentially many bi-partitions. \label{S3}
\end{enumerate}
}

By large amounts we mean that $\mathcal{C}$ scales in a similar way to its upper bound (for example if $\mathcal{C}$ is the entanglement monotone \emph{negativity} \cite{Neg}, it should scale faster than  $Polylog(n)$).  

  There is already strong evidence that  statement \ref{S1} is not valid for entanglement \cite{Datta2005} i.e it is known that  entanglement quantified by multiplicative negativity  at any point in a DQC1 circuit  is bounded from above by a constant.  Moreover, for universal quantum computing, it is known that \ref{S1}  can only be true for particular choices of entanglement measures \cite{Nest2012}. There is  some evidence to support  \ref{S1} for  other types of quantum correlations \cite{DiscordRMP}, i.e in DQC1 the correlations as measured by the operator Schmidt rank grow quickly with the system size \cite{Datta2007} and there is discord between the first qubit and the rest of the system \cite{DattaShajiCaves08}. However the latter statement does not imply that there are large amounts of discord for a bipartition that can support a lot of discord, moreover it is known that the results on discord  are not symmetric, i.e there is no discord if the measurements are made on the last $n$ qubits \cite{DiscordRMP}.

Our results below are related to statement \ref{S3}.  We present new tools to study entanglement in degenerate quantum systems undergoing unitary evolution  and  show that  while general DQC1 circuits can generate a state which is entangled over many bipartitions, this entanglement is always more sensitive to depolarizing noise than a generic quantum state.  We also show that specific DQC1 complete  circuits  have less entangling power than generic circuits.  As a conclusion of our results, we conjecture that statement \ref{S3} is probably  violated for entanglement (i.e for  DQC1 with large $\alpha$, see Eq. \ref{DQC1state} below).  The evidence in support of   statement  \ref{S3} for discord  is  significantly  stronger, since discord is more robust to depolarizing noise. 

\section{Definitions}

\subsection{ Entanglement and Discord}
 A state $\rho$ on 
$\mathcal{H}=\mathcal{H}_A\otimes\mathcal{H}_B$ is said to be separable with respect to the bipartition $\{A;B\}$ 
if and only if it can be decomposed as $\rho=\sum_l\rho^A_l\otimes\rho^B_l$ where $\rho^A_l$ are states on 
$\mathcal{H}_A$ and $\rho^B_l$ are states on $\mathcal{H}_B$. If $\rho$ is not separable, it is entangled.  
A sufficient condition for entanglement is that the partial transpose of $\rho$ denoted $(T\otimes \openone)(\rho)$ 
has a negative eigenvalue. Such a state is said to be a non positive partial transpose (non PPT) state; non PPT states are entangled.
If there are two  pure states $\ket{\psi},\ket{\phi}$ such that 
$\bra{\phi}(T\otimes \openone)(\rho)\ket{\phi}=0$ while $\bra{\psi}(T\otimes \openone)(\rho)\ket{\phi}\ne0$, then $\rho$ is non PPT \cite{TPNC2014}.  

While most quantum states  on $\mathcal{H}=\mathcal{H}_A\otimes\mathcal{H}_B$ are entangled when the Hilbert space dimensions are large, there is always a ball of separable states around the 
completely mixed state,
i.e for an $n$ qubit system and a given bipartition there is a finite $\epsilon$ such that all  $\tau$ with $||\tau-\frac{1}{2^n}\openone||_1<\epsilon$ are separable \cite{Gurvits2002}.  One  implication is that  for small $n$, a room temperature liquid state NMR processor is never entangled for any bipartition \cite{PhysRevLett.83.1054}.

States on the boundary of separable states are called \emph{boundary separable} \cite{TPNC2014}.  Here we define a new subset of boundary separable states (see appendix \ref{App:BSS} for proof that this is a subset) .

\begin{definition}\label{def:BSU}
A separable state $\rho$ is \emph{boundary separable in its unitary orbits} if  for all $\epsilon>0$ there is a unitary $U^\epsilon$ with $||U^\epsilon-\openone||<\epsilon$ such that $U^\epsilon \rho U^{\epsilon^\dagger}$ is entangled. 
\end{definition}

This subset is particularly relevant to systems undergoing unitary dynamics. Some states cannot be entangled by any unitary operation on the system \cite{Hildebrand2007}:
\begin{definition}\label{Def:SFS}
 A state $\rho$  is called \emph{separable from spectrum}  if for all unitaries $U$ the state $U\rho U^\dagger$ is separable. Similarly it is \emph{PPT from spectrum } if for all unitaries $U$ the state $U\rho U^\dagger$ is PPT.
\end{definition}

\emph{Discord} is an asymmetric measure of quantum correlations. A bipartite quantum state $\rho$ on $\mathcal{H}_A\otimes\mathcal{H}_B$  is  zero discord  with respect to a measurement on subsystem A if and only if \cite{OllivierZurek01,DiscordRMP,Boyer2016} there is a basis of states $\{\ket{l}\}$ for $\mathcal{H}_A$ such that $\rho=\sum_l a_l\,\ketbra{l}{l}^A\otimes \varrho^B_l$, with $\varrho^B_l$ states on $\mathcal{H}_B$. If $\rho$  is not zero discord, it is discordant.  For pure states, discord is symmetric and coincides with entanglement, moreover any entangled state is always discordant for measurements on either subsystem \cite{DiscordRMP}. Unlike separable states, the set of zero discord states is nowhere dense \cite{Ferraroetal10}, i.e there is no `ball' of zero discord states.

\subsection{DQC1}
The input to an $n+1$ qubit DQC1 circuit is a $\mathrm{Poly}(n)$  description of a unitary quantum circuit  $U$ that can be implemented efficiently as a $\mathrm{Poly}(n)$ sequence of one and two qubit gates chosen from a universal gate set. The circuit is a applied to the initial  $n+1$ qubit state 

\begin{equation}\label{DQC1state}
\rho^\alpha_n=\frac{1-\alpha}{2^n}\,\ketbra{0}{0}\otimes \openone_n+\frac{\alpha}{2^{n+1}} \openone_{n+1}.
\end{equation}
 
 After the evolution, the first qubit is measured and the expectation values of an operator $\sigma_\mu\in\{\sigma_y, \sigma_x\}$ on the first qubit is recorded \footnote{Our results  concern the intermediate states so they also apply to more general algorithms where the readout is not restricted.  Such versions of DQC1 have been considered in the past, although it is not  clear how far one can relax this restriction before making DQC1 universal},  i.e the output of the computation is 
\begin{equation}\label{eq:DQC1}
tr[U\rho^\alpha_nU^\dagger \sigma_{\mu}]=
\frac{1-\alpha}{2^n}tr\left[U\left(\ketbra{0}{0}\otimes\openone_n\right)U^\dagger\left(\sigma_{\mu}\otimes\openone_n\right)\right].
\end{equation}

Note that we define DQC1 as an estimation problem and use the term \emph{complete} below in that context. 

  We define cDQC1 to be the subset of DQC1 circuits with the restriction $U=[\,\ketbra{0}{0}\otimes\openone_n+\ket{1}\bra{1}\otimes V]H_1$ where $V$ is a unitary that can be efficiently decomposed into a polynomial number of one or two qubit gates and $H_1$ is a Hadamard gate on the first qubit.   This family of controlled unitaries  is vanishingly small in the set of all $n+1$ qubit unitaries. Nevertheless, it is sufficient for solving problems that are  DQC1 complete \cite{Shor2007} and most of the results regarding correlations in DQC1 were restricted to this   model.  To falsify  statements about necessary conditions for computational speedup\footnote{Assuming DQC1 is a SuQC}, it is sufficient to show that they do not hold for a single \emph{complete} problem for  DQC1, so the restriction to cDQC1 is well motivated \footnote{However, one should be careful since the reduction may require  additional qubits.}.

We will use the notation $\ncms$  to denote the subsystem consisting of the  $n$ qubits that are initially in the maximally mixed state. In cDQC1 the first qubit in the bipartition $\{1;\ncms\}$ plays a special role since it is the only qubit to have any coherence at any time. Note that $\{1;\ncms\}$ is a special case of  $\{1;n\}$, the set of  all bipartitions that have  one qubit for one party and $n$ for the rest (there are $n+1$ elements in the set). In general for any $k$ we  define the set of bipartitions $\{k,n+1-k\}$ where one party has $k$ qubits and the other has the rest.  There are overall $2^n-1$ non-trivial bipartitions. 

 \section{Entanglement in DQC1}

\subsection{The pure, $\alpha=0$ case:}
 Knill and Laflamme \cite{KL} intoduced DQC1 as an algorithm for liquid state NMR where it is possible to efficiently initialize $\rho^\alpha$ with  $\alpha$  close to 1. It is however instructive to consider entanglement  in the idealized  $\alpha=0$ case before continuing to general $\alpha$.

\begin{lemma}\label{lemma:BSU}
The  state $\rho^0_n$ is boundary separable in its unitary orbit for all $2^n -1$ bipartitions $\{k;n+1-k\}$ where $0<k\le n$.
\end{lemma}
\begin{proof} To find $U^\epsilon$ we start with $n=1$.  Let $R_\theta$ be the unitary  that acts  trivially on the subspace $\mathrm{Span}(\ket{01},\ket{10})^\perp$ and  s.t. $R_\theta\ket{01}=\cos\theta\ket{01}+\sin\theta\ket{10}$ and $R_\theta\ket{10}=-\sin\theta\ket{01}+\cos\theta\ket{10}$; then,  for all $0<\theta<\frac{\pi}{2}$,  $R_\theta \rho^0_1R_\theta^\dagger$ is entangled since it is non PPT. That can be seen by noting that   $\bra{11}(T\otimes\openone)(R_\theta\rho^0_1R_\theta^\dagger)\ket{11}=0$ whereas $\bra{00}(T\otimes\openone)(R_\theta\rho^0_1 R_\theta^\dagger)\ket{11}=\frac{1}{2}\cos\theta\sin\theta\ne0$.
Then, since $R_\theta =\openone$ for $\theta=0$, for any $\epsilon>0$ there is $0<\theta<\frac{\pi}{2}$ such that $\|R_\theta -\openone\| < \epsilon$ due to the continuity of $R_\theta$ in $\theta$,
and we can take $U^\epsilon=R_\theta$. This is also true for the more general $n$ and any bipartition; for example if the first qubit is in part $A$ and the $l_{th}$ qubit is in part $B$ it is possible to have $R_\theta$ act on those two qubits and entangle them, the rest of the system will remain in a factorized maximally mixed state.  
\end{proof}
 What is rather surprising is that the subclass of unitaries used in cDQC1 are precisely those that do not produce entanglement in the $\{1;\ncms\}$  bipartition \cite{DiscordRMP}. 

So,  on the one hand, entanglement in the $\{1;\ncms\}$  bipartition is  easy in a generic DQC1 circuit (at $\alpha=0$); on the other hand  there is a  subclass, cDQC1, which is known to contain DQC1 complete problems and cannot generate entanglement in this cut. 
Our main result from the analysis of $\alpha=0$ is that the the existence of entanglement  in the general case does not indicate  that entanglement should exist in subsets of circuits  that can encode DQC1 complete problems. 

\subsection{  The $\{1;n\}$ bipartition ($\alpha > 0$):}

We continue with  entanglement  at $1 >\alpha>0$ and  a  general  $1$\,qubit -- $n$\,qubits  bipartition $\{1;n\}$  (of which $\{1;\ncms\}$  is a special case).

\begin{lemma} \label{lemma:1n}
A DQC1 circuit cannot generate entanglement at  any $\{1;n\}$ bipartition if and only if  $\alpha\ge\frac{1}{2}$.
\end{lemma}
\begin{proof}
  In  \cite{Johnston2013}, Johnston showed that: given an $n+1$ qubit state with  eigenvalue spectrum $\lambda_1\ge\lambda_2\ge...\lambda_{2^{n+1}}$, the following is a necessary and sufficient condition for separability from spectrum in any $\{1;n\}$ bipartition, 
\begin{equation} \label{eq:sfs}
\lambda_1\le \lambda_{2^{n+1}-1}+2\sqrt{\lambda_{2^{n+1}-2}\lambda_{2^{n+1}}}.
\end{equation}

The DQC1 state has a degenerate spectrum with two eigenvalues $\frac{2-\alpha}{2^{n+1}}$ and $\frac{\alpha}{2^{n+1}}$ each with degeneracy $2^n$.   So condition \ref{eq:sfs} is violated  for $\alpha \ge \frac{1}{2}$.
\end{proof}

Lemma \ref{lemma:1n} is surprising since  $||U \rho_n^{\alpha}U^\dagger-\frac{1}{2^{n+1}}\openone||_1=1-\alpha$  so $U\rho_n^{0.5} U^\dagger$ is separable, but far outside the ball of separable states at large $n$.  To see this,  take the state $\tau=\sum_{i=2}^{2^{n+1}-1}\frac{1}{2^{n+1}-2} \ketbra{i}{i}$  with $ ||\tau-\frac{1}{2^{n+1}}\openone||=\frac{2}{2^n}$. Since $\tau$ has eigenvalues $\lambda_{2^{n+1}} = \lambda_{2^{n+1}-1} = 0$  then, by eq  \eqref{eq:sfs}, for any $\{1;n\}$  bipartition there is some $U$ such that $U\tau U^\dagger$ is entangled and therefore outside the ball of separable states.  This result shows that entanglement in the set $U\rho^\alpha_n U^\dagger$  is particularly sensitive to noise in the initial state. 

 The result above  complements the result of Datta, Flammia, and Caves \cite{Datta2005} who found an explicit family of unitaries  such that for $\alpha < \frac{1}{2}$ the state  $U\rho^\alpha_n U^\dagger$ is entangled for any bipartition. Furthermore these states are not PPT which is   consistent with  evidence  that PPT from spectrum is the same as separable from spectrum \cite{Arunachalam2014}.

\subsection{General bipartitions ($\alpha\ge 0$):}

Moving to the more general case, we build on the results of  Hildebrand  \cite{Hildebrand2007} who provided a necessary and sufficient condition for a state to be PPT from spectrum.  Hildebrand's  general condition  (see Theorem \ref{main} below) is generally difficult to apply to states with a generic  eigenvalue spectrum. In appendix \ref{App:degenerate} we show how to apply this result to the highly degenerate states  in the set $\mathcal{S}^{{\lambda_+},{\lambda_-}}_{m,n}$ defined in Def. \ref{def:S}.  The DQC1 states are in this set and the following is a special case of lemma \ref{Lemma:Hild}:

\begin{lemma}\label{lemma5} Let $\rho$ be a state on $\mathcal{H}_k\otimes \mathcal{H}_{n-k+1}$ with $k< \frac{n+1}{2}$. If $\rho$ has  two eigenvalues $\lambda_1>\lambda_2$ each with degeneracy $2^{n}$ then    $\rho$  is PPT from spectrum if and only if $\frac{1}{2}(\lambda_1+\lambda_2)-{2^{k-1}}(\lambda_1-\lambda_2)\ge0$    \end{lemma}

\begin{proof}
Using the notation defined in  appendix \ref{App:Hild}, we have  $p=2^k$, $p_+ = 2^{k-1}(2^k-1) \leq 2^n$ and of course $p_- \leq 2^n$ so that $\rho\in\mathcal{S}^{{\lambda_2},{\lambda_1}}_{k,n-k+1}$ and lemma \ref{Lemma:Hild} applies.
\end{proof}

\begin{lemma} A necessary condition, and a sufficient condition, for  $\rho^\alpha_n$ to be PPT from spectrum for all bipartitions
of the $n+1$ qubits $(n \geq 2)$ are respectively 
\begin{equation}
\alpha \geq 1 - \frac{1}{2^{\lfloor \frac{n}{2}\rfloor}},
\quad\text{and}\quad\alpha \geq 1 - \frac{1}{2^{\lfloor \frac{n+1}{2}\rfloor}}
\end{equation}
where $\lfloor x \rfloor$ is the floor of $x$ (i.e $x$ rounded down to the nearest integer).
\end{lemma}
\begin{proof}
The two eigenvalues   of $\rho^\alpha_n$ are  $\lambda_1=\frac{2-\alpha}{2^{n+1}}$  and $\lambda_2=\frac{\alpha}{2^{n+1}}$. By lemma \ref{lemma5},  
 $\rho^\alpha_n$ is  PPT from spectrum
for any  $\{ k;(n+1-k)\}$ cut  with $1\leq k < \frac{n+1}{2}$   if an only if 
\begin{equation} \label{eq:APPT}
\alpha \ge 1-\frac{1}{2^k}.
\end{equation}
If  $n=2m$,  
$k < \frac{n+1}{2}$ if and only if $k\leq m = \lfloor \frac{n}{2}\rfloor = \lfloor \frac{n+1}{2}\rfloor$
and the conditions coincide, giving a necessary and sufficient condition.

 If $n=2m+1$,  $\{m+1;m+1\}$ bipartitions are to be handled separately. 
 If $U\rho^\alpha_n U^\dagger$ is not PPT for a $\{m+1;m+1\}$ bipartition, then 
 $(U\otimes\openone)\rho^\alpha_{n+1}(U\otimes\openone)^\dagger$ is not PPT for a $\{m+1;m+2\}$ bipartition; consequently, if 
 $\alpha\geq 1 - \frac{1}{2^{\lfloor \frac{n+1}{2}\rfloor}} = 1-\frac{1}{2^k}$ for $k=m+1$, then 
 by \eqref{eq:APPT}, $\rho^\alpha_{n+1}$ is
 PPT from spectrum for $\{m+1;m+2\}$ bipartition and thus  $\rho^\alpha_n$ is PPT from spectrum also for $\{m+1;m+1\}$ bipartitions;
that proves the sufficiency condition.
 
If $n=2m+1$ and if   $U\rho^\alpha_{n-1}U^\dagger$ is not PPT for a
$\{m; m+1\}$ bipartition, then 
$(U\otimes\openone)\rho^\alpha_n(U\otimes\openone)^\dagger$ is not PPT for
a $\{m+1;m+1\}$ bipartition; thus, if $\rho^\alpha_n$ is PPT from spectrum for $\{m+1;m+1\}$ bipartitions, then $\rho^\alpha_{n-1}$ is PPT
from spectrum for $\{m; m+1\}$ bipartitions, implying by equation \eqref{eq:APPT} with $k=m =\lfloor \frac{n}{2}\rfloor$ that
$\alpha \geq  1-\frac{1}{2^{\lfloor\frac{n}{2}\rfloor}}$.%
%
\end{proof}

The scaling of this condition means that for any fixed $\alpha$, or even for $\alpha <  1- \frac{1}{\mathrm{Poly}(n)}$ there will always be states that are entangled for some $U$ and large enough $n$. Moreover the number of bipartitions for which this statement holds grows exponentially with the size of $n$. In a follow up paper \cite{BBMinprep} we show how to  construct an explicit family of  $U$ such that $U \rho^\alpha_n U^\dagger$  is not PPT  when the condition \eqref{eq:APPT} is violated. 

\subsection{ Discord}

The  lack of entanglement for low $n$ in liquid state NMR experiments and DQC1 led to various conjectures about discord as a more appropriate signature of the quantum advantage \cite{arXiv:quant-ph/0110029,arXiv:0906.3656,arXiv:1109.5549}. Datta et al. \cite{DattaShajiCaves08} provided evidence for this conjecture by showing that the separable $\{1;\ncms\}$  bipartite cut in cDQC1 was usually discordant with respect to a measurement on the first qubit. This state is, however, never discordant with respect to a measurement on  $\ncms$   \cite{DiscordRMP}.  Moreover there are also (seemingly DQC1 complete)  subsets of cDQC1 where the final state is not discordant for a measurement on the first qubit  \cite{DakicVedralBrukner10}, however, even in these circuits the state may be discordant at some intermediate time \cite{DiscordRMP}. 

For a qualitative study of  the role  of discord in DQC1  it is enough to study the `clean' case $\alpha=0$. This follows from the fact that a $n+1$ qubit state of the form  $(1-\alpha)\rho+\alpha\frac{\openone}{2^{n+1}}$ is discordant if and only if $\rho$ is discordant \cite{Groisman2007}. This fact, together with the fact that entanglement implies discord means that for any bipartition and any $\alpha$ there are unitaries $U$ such that $U\rho^\alpha_n U^\dagger$ is discordant with respect to a measurement on either subsystem and any other bipartition.  

\section{Conclusions}

Questions regarding the role of entanglement in quantum computing algorithms have been studied since the first quantum algorithms were tested in liquid state NMR.  The DQC1 algorithm was designed as a testbed for answering these questions, but even with this simplified model the results are inconclusive.  Here we studied the ability of DQC1 to generate entanglement under various constraints. 

 Noise in the initial state (i.e $\alpha$ in Eq. \ref{DQC1state} ) determines the ability of a generic circuit to generate entanglement. We showed that in any  $\{1;n\}$ bipartition, a circuit cannot generate entanglement when $\alpha\ge\frac{1}{2}$ (Lemma \ref{lemma:1n}). We also provided a necessary and sufficient condition for the DQC1 circuit to generate non PPT states (Eq. \ref{eq:APPT}).

We defined a new property called {boundary separable in unitary orbits} (def. \ref{def:BSU}) and showed that the initial states of DQC1 at $\alpha=0$ are easy to entangle in any  bipartition (lemma \ref{lemma:BSU}). On the other hand the  DQC1-complete subset, cDQC1 cannot generate entanglement in the $\{1;\ncms\}$ cut, despite the fact that the first qubit is the only one with any coherence at any time. We conclude that there is no reason to suspect that families of  DQC1-complete circuits are those that generate more entanglement than other non-trivial families of circuits. 

 Our most surprising result is that the  entanglement in the set all of states generated in a DQC1 circuit is  more fragile  to depolarizing noise than a generic mixed state.   Based on this and the conclusion above, we are optimistic about the possibility of finding DQC1-complete circuits where entanglement is never generated at any point for the vast majority of bipartite cuts (and possibly all bipartite cuts) at all $n$. Such circuits  would provide extremely strong evidence that quantum computational speedup can be achieved without entanglement. 
 
 We pointed out that, at the qualitative level, the study of discord can be restricted to $\alpha=0$. Based on that, we conclude that  a subset of circuits that do not generate discord on the one hand, and are DQC1 complete on the other is unlikely to exist. 

We believe that the  next challenge will be to  find a  family  of DQC1  circuits which   encode classically hard  computational problems, and at the same time, do not generate any entanglement for most bipartitions at some $\alpha<1-\frac{1}{\mathrm{Poly}({n})}$;  or conversely  give  strong evidence that such a family is unlikely to exist, for example by finding  an algorithm that can simulate any separable instance of DQC1. 

Our results were presented in the context of DQC1. The  approaches developed here will be useful for further study of entanglement  and quantum correlations in other candidate SuQCs.  More generally, our methods can be applied to the study of entanglement and quantum correlation in other complex systems involving mixed quantum states.

\appendix

\section{PPT from spectrum for DQC1 states}\label{App:Hild}

In this section we provide some general results regarding the possibility that the state $U \rho^\alpha_n U^\dagger$ is not PPT. We begin with a review of results by Hildebrand \cite{Hildebrand2007} and continue to explicitly calculate the necessary and sufficient conditions for PPT from spectrum for a particular set of states where the largest and smallest eigenvalues are highly degenerate. 
 
\subsection{Recap of  definitions and main theorem from Hildebrand 2007  \cite{Hildebrand2007}}

\subsubsection{Notations and definitions}
In the following, $[n]$ denotes the set $\{1,\ldots, n\}$ of $n$ elements;
$\mathrm{H}(n)$ denotes the space of $n\times n$ Hermitian matrices
or Hermitian operators on $\sh_n$, 
$\mathrm{H}_+(nm)$ denotes the set of positive semidefinite $nm\times nm$ matrices
or positive semi definite (PSD) operators on $\sh_n\otimes \sh_m$; 


\begin{definition}
 Let  $p_+ = \displaystyle\frac{p(p+1)}{2}$ and $p_- = \displaystyle\frac{p(p-1)}{2}$
 for $p\in\mathbb{N}$; let also $S^p_+ = \{(i,j) \mid 1\leq i\leq j\leq p\}$ and $S^p_- = \{(i,j)\mid 1\leq i < j\leq p\}$. A \emph{linear ordering} of the pairs $(i,j) \in S^p_+$ (resp in $S^p_-$) is a bijective map $\sigma_+:S^p_+ \to [p_+]$
(resp $\sigma_-:S^p_- \to [p_-]$).
\end{definition}
\begin{definition}
The linear ordering $\sigma_-:S_- \to [p_-]$ is said to be \emph{consistent} with $\sigma_+:S^p_+\to [p_+]$  if for all $(k_1,l_1), (k_2,l_2) \in S_-^p$, 
$\sigma_+(k_1,l_1) < \sigma_+(k_2,l_2)$ implies $\sigma_-(k_1,l_1) < \sigma_-(k_2,l_2)$.
\end{definition}

\begin{remark}
For each $\sigma_+:S^p_+\to [p_+]$ there is exactly one 
$\sigma_-:S^p_- \to [p_-]$ that is consistent with $\sigma_+$.
\end{remark}
\begin{definition}
 Let $x\in\mathbb{R}^p$ 
be a vector with non negative entries. A \emph{linear ordering} 
$(\sigma_+,\sigma_-)$ is said to be \emph{compatible with} $x$ if $\sigma_+(k_1,l_1) < \sigma_+(k_2,l_2)$ implies $x_{k_1}x_{l_1} \geq x_{k_2}x_{l_2}$.
\end{definition}

The linear ordering above is a simple way to put the  products
$x_kx_l$ in decreasing order for $1\leq k\leq l \leq p$. 
If the products are all distinct, there is just one way. In case of identical elements the order is not relevant. 

\begin{align*} 
\Sigma_{\pm} = \Big\{ (\sigma_+,\sigma_-) \mid &\ \exists\ x \in \mathbb{R}^p \mid x_1 > x_2 > \cdots > x_p > 0:\\
 & (\sigma_+,\sigma_-) \text{ compatible with $x$} \Big\}.
\end{align*}

\begin{definition}\label{def:lambda}
If $\lambda = (\lambda_i)_{1\leq i \leq nm}$ is a sorted list of $mn$ real 
numbers in decreasing order, $p=\min(m,n)$ and $(\sigma_+,\sigma_-)$
is a consistent  pair of orderings of $S^p_+$ and $S^p_-$, then $\Lambda(\lambda;\sigma_+,
\sigma_-)$ is the $p\times p$ matrix defined by
\[
\Lambda(\lambda;\sigma_+,
\sigma_-) = \begin{cases} \lambda_{nm+1-\sigma_+(k,l)} & k \leq l\\[1.5ex]
 -\lambda_{\sigma_-(l,k)} & k > l
 \end{cases}
\]
\end{definition}

\subsubsection{Main result}

\begin{theorem}\label{main} If $A\in \mathrm{H}_+(nm)$ has  
$\lambda = (\lambda_i)_{1\leq i \leq nm}$ as eigenvalues
in decreasing
order,
then $A$  has a positive semi-definite partial transpose (PPT) for all decompositions of $\mathcal{H}_{nm}$ as a tensor product space $\mathcal{H}_n \otimes
\mathcal{H}_m$ if and only if for all $(\sigma_+,\sigma_-) \in \Sigma_{\pm}$ the following holds:
\[ \Lambda(\sigma_+,\sigma_-) + \Lambda(\sigma_+,\sigma_-)^T \succeq 0.
\] 
\end{theorem}
\begin{proof} Cf. \cite[Theorem 1]{Hildebrand2007} \end{proof}

\subsection{Application to special states with highly degenerate eigenvalue spectrum}\label{App:degenerate}

\begin{definition} \label{def:S}
For $m < n$, we define   $\mathcal{S}^{{\lambda_+},{\lambda_-}}_{m,n}$ as the set of  states  $\tau \in H_+(2^m2^n)$  that 
 have ${\lambda_+}$ as their largest eigenvalue with degeneracy at least $p_-=2^{2m-1}-2^{m-1}$ and ${\lambda_-}$ as the smallest eigenvalue with degeneracy at least $p_+=2^{2m-1}+2^{m-1}$ (note that there are $2^{m+n}-2^{2m}$ free eigenvalues).  
\end{definition}

\begin{lemma}\label{Lemma:Hild}
All states $\tau \in \mathcal{S}^{{\lambda_+},{\lambda_-}}_{m,n}$ are separable from spectrum if and only if $\frac{1}{2}({\lambda_+}+{\lambda_-})-2^{m-1}({\lambda_+}-{\lambda_-})\ge0$
\end{lemma}

\begin{proof}Let $\lambda$ be the list of eigenvalues of $\tau \in \mathcal{S}^{{\lambda_+},{\lambda_-}}_{m,n}$ sorted in decreasing order and  $p=2^m$;   all the matricies $\Lambda(\lambda;\sigma_+,\sigma_-)$ are 
then equal and have the following form 
\[
\Lambda_{k,l} = \begin{cases}
 \displaystyle\phantom{-} {\lambda_-} &\hspace*{1cm} 1 \leq k \leq l \leq p \\[2ex]
 \displaystyle -{\lambda_+}           &\hspace*{1cm} p \geq k > l \geq 1
                                                 \end{cases}
\]
If follows that the off diagonal entries of $\Lambda + \Lambda^T$ are equal to $\lambda_- - \lambda_+$ and the diagonal entries are $2\lambda_-$ and thus
\[
\Lambda + \Lambda^T = (\lambda_-+\lambda_+) \openone  + (\lambda_- - \lambda_+)K  
\]
where $K_{i,j} = 1$ for all $1 \leq i,j \leq p$. For any $\bm{x}\in \mathbb{R}^p$
\[ \bm{x}^T(\Lambda + \Lambda^T)\bm{x} =  (\lambda_++\lambda_-) \|\bm{x}\|^2 - (\lambda_+ - \lambda_-) \big(\sum_i x_i\big)^2 
\]
%
%
%
and the minimum value for $\|\bm{x}\|=1$ is obtained if all $x_i$ are equal, i.e. $x_i=\frac{1}{\sqrt{p}}$ for $1\leq i\leq p$, giving
a minimum of $(\lambda_++\lambda_-) -p(\lambda_+-\lambda_-)$. By theorem \ref{main}, $\tau$ is separable from spectrum if and only if $(\lambda_++\lambda_-) -p(\lambda_+-\lambda_-)\geq 0$.
\end{proof}

\section{Boundary separable states}\label{App:BSS}
In the main text we stated that the set of states that are boundary separable in their unitary orbits are a subset of the set of boundary separable states. Below is proof of that statement. 
\begin{lemma}
If $\rho$ is boundary separable in its unitary orbits, then $\rho$ is boundary separable
\end{lemma}
\begin{proof}
The following inequalities apply whether $\|\cdot\|$ is the trace norm, the operator norm or the Hilbert-Schmidt norm. 
Here  we assume it is the trace norm.
Let $U$ be such that $\|U-\openone\| < \epsilon/2$ and $U\rho U^\dagger$ is entangled. Using the fact
that $\|\cdot\|$ is a norm (i.e the triangle inequality holds), that $\|A^\dagger\| = \|A\|$, $\|AB\| \leq \|A\|\|B\|$ and $\|AB\| \leq \|A\|_\infty\|B\|$, 
\begin{align*}
\|U\rho U^\dagger-\rho\| &= \|U\rho (U^\dagger -\openone) + (U - \openone)\rho\|\\
&\leq \|U\|_\infty \|\rho\| \|U^\dagger - \openone\| + \|U -\openone\| \|\rho\|\\
&< \epsilon \qedhere
\end{align*}
\end{proof}

\vspace{10pt}

\begin{acknowledgements}
We would like to thank N. Johnston, R. Laflamme, and R. Liss for useful discussions and comments.  A.B was partly funded through NSERC, Industry Canada and CIFAR. M.B. was partly support by NSERC and FQRNT through 
INTRIQ. T.M was funded by  the Wolfson Foundation and the Israeli MOD
Research and Technology Unit.  A.B and T.M were partly supported The Gerald Schwartz \& Heather Reisman Foundation. A.B is currently at the Center for Quantum Information and Quantum Control at the University of Toronto.

\end{acknowledgements}
 
%



\begin{thebibliography}{42}%
\makeatletter
\providecommand \@ifxundefined [1]{%
 \@ifx{#1\undefined}
}%
\providecommand \@ifnum [1]{%
 \ifnum #1\expandafter \@firstoftwo
 \else \expandafter \@secondoftwo
 \fi
}%
\providecommand \@ifx [1]{%
 \ifx #1\expandafter \@firstoftwo
 \else \expandafter \@secondoftwo
 \fi
}%
\providecommand \natexlab [1]{#1}%
\providecommand \enquote  [1]{``#1''}%
\providecommand \bibnamefont  [1]{#1}%
\providecommand \bibfnamefont [1]{#1}%
\providecommand \citenamefont [1]{#1}%
\providecommand \href@noop [0]{\@secondoftwo}%
\providecommand \href [0]{\begingroup \@sanitize@url \@href}%
\providecommand \@href[1]{\@@startlink{#1}\@@href}%
\providecommand \@@href[1]{\endgroup#1\@@endlink}%
\providecommand \@sanitize@url [0]{\catcode `\\12\catcode `\$12\catcode
  `\&12\catcode `\#12\catcode `\^12\catcode `\_12\catcode `\%12\relax}%
\providecommand \@@startlink[1]{}%
\providecommand \@@endlink[0]{}%
\providecommand \url  [0]{\begingroup\@sanitize@url \@url }%
\providecommand \@url [1]{\endgroup\@href {#1}{\urlprefix }}%
\providecommand \urlprefix  [0]{URL }%
\providecommand \Eprint [0]{\href }%
\providecommand \doibase [0]{http://dx.doi.org/}%
\providecommand \selectlanguage [0]{\@gobble}%
\providecommand \bibinfo  [0]{\@secondoftwo}%
\providecommand \bibfield  [0]{\@secondoftwo}%
\providecommand \translation [1]{[#1]}%
\providecommand \BibitemOpen [0]{}%
\providecommand \bibitemStop [0]{}%
\providecommand \bibitemNoStop [0]{.\EOS\space}%
\providecommand \EOS [0]{\spacefactor3000\relax}%
\providecommand \BibitemShut  [1]{\csname bibitem#1\endcsname}%
\let\auto@bib@innerbib\@empty
\bibitem [{\citenamefont {Nielsen}\ and\ \citenamefont
  {Chuang}(2010)}]{NielsenandChuang}%
  \BibitemOpen
  \bibfield  {author} {\bibinfo {author} {\bibfnamefont {M.}~\bibnamefont
  {Nielsen}}\ and\ \bibinfo {author} {\bibfnamefont {I.}~\bibnamefont
  {Chuang}},\ }\href@noop {} {\emph {\bibinfo {title} {{Quantum Computation and
  Quantum Information}}}}\ (\bibinfo  {publisher} {Cambridge University
  Press},\ \bibinfo {year} {2010})\BibitemShut {NoStop}%
\bibitem [{\citenamefont {Steane}(1998)}]{Steane1998}%
  \BibitemOpen
  \bibfield  {author} {\bibinfo {author} {\bibfnamefont {A.}~\bibnamefont
  {Steane}},\ }\href {\doibase 10.1088/0034-4885/61/2/002} {\bibfield
  {journal} {\bibinfo  {journal} {Rep. Prog.  Phys.}\ }\textbf
  {\bibinfo {volume} {61}},\ \bibinfo {pages} {117} (\bibinfo {year}
  {1998})}\BibitemShut {NoStop}%
\bibitem [{\citenamefont {Bernstein}\ \emph {et~al.}(1997)\citenamefont
  {Bernstein}, \citenamefont {Bennett}, \citenamefont {Brassard},\ and\
  \citenamefont {Vazirani}}]{Bernstein1997}%
  \BibitemOpen
  \bibfield  {author} {\bibinfo {author} {\bibfnamefont {E.}~\bibnamefont
  {Bernstein}}, \bibinfo {author} {\bibfnamefont {C.~H.}\ \bibnamefont
  {Bennett}}, \bibinfo {author} {\bibfnamefont {G.}~\bibnamefont {Brassard}}, \
  and\ \bibinfo {author} {\bibfnamefont {U.~V.}\ \bibnamefont {Vazirani}},\
  }\href {\doibase 10.1137/S0097539796300933} {\bibfield  {journal} {\bibinfo
  {journal} {SIAM J. Comp.}\ }\textbf {\bibinfo {volume} {26}},\
  \bibinfo {pages} {1510} (\bibinfo {year} {1997})},\ \Eprint
  {http://arxiv.org/abs/quant-ph/9701001} {arXiv:quant-ph/9701001} \BibitemShut
  {NoStop}%
\bibitem [{\citenamefont {Kelly}\ \emph {et~al.}(2015)\citenamefont {Kelly},
  \citenamefont {Barends}, \citenamefont {Fowler}, \citenamefont {Megrant},
  \citenamefont {Jeffrey}, \citenamefont {White}, \citenamefont {Sank},
  \citenamefont {Mutus}, \citenamefont {Campbell}, \citenamefont {Chen},
  \citenamefont {Chen}, \citenamefont {Chiaro}, \citenamefont {Dunsworth},
  \citenamefont {Hoi}, \citenamefont {Neill}, \citenamefont {O'Malley},
  \citenamefont {Quintana}, \citenamefont {Roushan}, \citenamefont
  {Vainsencher}, \citenamefont {Wenner}, \citenamefont {Cleland},\ and\
  \citenamefont {Martinis}}]{SCerror}%
  \BibitemOpen
  \bibfield  {author} {\bibinfo {author} {\bibfnamefont {J.}~\bibnamefont
  {Kelly}}, \bibinfo {author} {\bibfnamefont {R.}~\bibnamefont {Barends}},
  \bibinfo {author} {\bibfnamefont {A.~G.}\ \bibnamefont {Fowler}}, \bibinfo
  {author} {\bibfnamefont {A.}~\bibnamefont {Megrant}}, \bibinfo {author}
  {\bibfnamefont {E.}~\bibnamefont {Jeffrey}}, \bibinfo {author} {\bibfnamefont
  {T.~C.}\ \bibnamefont {White}}, \bibinfo {author} {\bibfnamefont
  {D.}~\bibnamefont {Sank}}, \bibinfo {author} {\bibfnamefont {J.~Y.}\
  \bibnamefont {Mutus}}, \bibinfo {author} {\bibfnamefont {B.}~\bibnamefont
  {Campbell}}, \bibinfo {author} {\bibfnamefont {Y.}~\bibnamefont {Chen}},
  \bibinfo {author} {\bibfnamefont {Z.}~\bibnamefont {Chen}}, \bibinfo {author}
  {\bibfnamefont {B.}~\bibnamefont {Chiaro}}, \bibinfo {author} {\bibfnamefont
  {A.}~\bibnamefont {Dunsworth}}, \bibinfo {author} {\bibfnamefont {I.-C.}\
  \bibnamefont {Hoi}}, \bibinfo {author} {\bibfnamefont {C.}~\bibnamefont
  {Neill}}, \bibinfo {author} {\bibfnamefont {P.~J.~J.}\ \bibnamefont
  {O'Malley}}, \bibinfo {author} {\bibfnamefont {C.}~\bibnamefont {Quintana}},
  \bibinfo {author} {\bibfnamefont {P.}~\bibnamefont {Roushan}}, \bibinfo
  {author} {\bibfnamefont {A.}~\bibnamefont {Vainsencher}}, \bibinfo {author}
  {\bibfnamefont {J.}~\bibnamefont {Wenner}}, \bibinfo {author} {\bibfnamefont
  {A.~N.}\ \bibnamefont {Cleland}}, \ and\ \bibinfo {author} {\bibfnamefont
  {J.~M.}\ \bibnamefont {Martinis}},\ }\href {\doibase 10.1038/nature14270}
  {\bibfield  {journal} {\bibinfo  {journal} {Nature}\ }\textbf {\bibinfo
  {volume} {519}},\ \bibinfo {pages} {66} (\bibinfo {year} {2015})},\ \Eprint
  {http://arxiv.org/abs/1411.7403} {arXiv:1411.7403 [quant-ph]} \BibitemShut
  {NoStop}%
\bibitem [{\citenamefont {Negrevergne}\ \emph {et~al.}(2006)\citenamefont
  {Negrevergne}, \citenamefont {Mahesh}, \citenamefont {Ryan}, \citenamefont
  {Ditty}, \citenamefont {Cyr-Racine}, \citenamefont {Power}, \citenamefont
  {Boulant}, \citenamefont {Havel}, \citenamefont {Cory},\ and\ \citenamefont
  {Laflamme}}]{Negrevergne2006}%
  \BibitemOpen
  \bibfield  {author} {\bibinfo {author} {\bibfnamefont {C.}~\bibnamefont
  {Negrevergne}}, \bibinfo {author} {\bibfnamefont {T.~S.}\ \bibnamefont
  {Mahesh}}, \bibinfo {author} {\bibfnamefont {C.~A.}\ \bibnamefont {Ryan}},
  \bibinfo {author} {\bibfnamefont {M.}~\bibnamefont {Ditty}}, \bibinfo
  {author} {\bibfnamefont {F.}~\bibnamefont {Cyr-Racine}}, \bibinfo {author}
  {\bibfnamefont {W.}~\bibnamefont {Power}}, \bibinfo {author} {\bibfnamefont
  {N.}~\bibnamefont {Boulant}}, \bibinfo {author} {\bibfnamefont
  {T.}~\bibnamefont {Havel}}, \bibinfo {author} {\bibfnamefont {D.~G.}\
  \bibnamefont {Cory}}, \ and\ \bibinfo {author} {\bibfnamefont
  {R.}~\bibnamefont {Laflamme}},\ }\href
  {http://prl.aps.org/abstract/PRL/v96/i17/e170501} {\bibfield  {journal}
  {\bibinfo  {journal} {Phys. Rev. Let.}\ }\textbf {\bibinfo {volume}
  {96}},\ \bibinfo {pages} {170501} (\bibinfo {year} {2006})}\BibitemShut
  {NoStop}%
\bibitem [{\citenamefont {Lu}\ \emph {et~al.}(2015)\citenamefont {Lu},
  \citenamefont {Li}, \citenamefont {Trottier}, \citenamefont {Li},
  \citenamefont {Brodutch}, \citenamefont {Krismanich}, \citenamefont
  {Ghavami}, \citenamefont {Dmitrienko}, \citenamefont {Long}, \citenamefont
  {Baugh},\ and\ \citenamefont {Laflamme}}]{Lu}%
  \BibitemOpen
  \bibfield  {author} {\bibinfo {author} {\bibfnamefont {D.}~\bibnamefont
  {Lu}}, \bibinfo {author} {\bibfnamefont {H.}~\bibnamefont {Li}}, \bibinfo
  {author} {\bibfnamefont {D.-A.}\ \bibnamefont {Trottier}}, \bibinfo {author}
  {\bibfnamefont {J.}~\bibnamefont {Li}}, \bibinfo {author} {\bibfnamefont
  {A.}~\bibnamefont {Brodutch}}, \bibinfo {author} {\bibfnamefont {A.~P.}\
  \bibnamefont {Krismanich}}, \bibinfo {author} {\bibfnamefont
  {A.}~\bibnamefont {Ghavami}}, \bibinfo {author} {\bibfnamefont {G.~I.}\
  \bibnamefont {Dmitrienko}}, \bibinfo {author} {\bibfnamefont
  {G.}~\bibnamefont {Long}}, \bibinfo {author} {\bibfnamefont {J.}~\bibnamefont
  {Baugh}}, \ and\ \bibinfo {author} {\bibfnamefont {R.}~\bibnamefont
  {Laflamme}},\ }\href {\doibase 10.1103/PhysRevLett.114.140505} {\bibfield
  {journal} {\bibinfo  {journal} {Phys. Rev. Lett.}\ }\textbf {\bibinfo
  {volume} {114}},\ \bibinfo {pages} {140505} (\bibinfo {year}
  {2015})}\BibitemShut {NoStop}%
\bibitem [{\citenamefont {Monz}\ \emph {et~al.}(2011)\citenamefont {Monz},
  \citenamefont {Schindler}, \citenamefont {Barreiro}, \citenamefont {Chwalla},
  \citenamefont {Nigg}, \citenamefont {Coish}, \citenamefont {Harlander},
  \citenamefont {H{\"{a}}nsel}, \citenamefont {Hennrich},\ and\ \citenamefont
  {Blatt}}]{Monz2011}%
  \BibitemOpen
  \bibfield  {author} {\bibinfo {author} {\bibfnamefont {T.}~\bibnamefont
  {Monz}}, \bibinfo {author} {\bibfnamefont {P.}~\bibnamefont {Schindler}},
  \bibinfo {author} {\bibfnamefont {J.~T.}\ \bibnamefont {Barreiro}}, \bibinfo
  {author} {\bibfnamefont {M.}~\bibnamefont {Chwalla}}, \bibinfo {author}
  {\bibfnamefont {D.}~\bibnamefont {Nigg}}, \bibinfo {author} {\bibfnamefont
  {W.~A.}\ \bibnamefont {Coish}}, \bibinfo {author} {\bibfnamefont
  {M.}~\bibnamefont {Harlander}}, \bibinfo {author} {\bibfnamefont
  {W.}~\bibnamefont {H{\"{a}}nsel}}, \bibinfo {author} {\bibfnamefont
  {M.}~\bibnamefont {Hennrich}}, \ and\ \bibinfo {author} {\bibfnamefont
  {R.}~\bibnamefont {Blatt}},\ }\href {\doibase 10.1103/PhysRevLett.106.130506}
  {\bibfield  {journal} {\bibinfo  {journal} {Phys. Rev. Lett.}\
  }\textbf {\bibinfo {volume} {106}},\ \bibinfo {pages} {130506} (\bibinfo
  {year} {2011})}\BibitemShut {NoStop}%
  \bibitem{Preskill2012}
  \BibitemOpen
  \bibfield  {author} {\bibinfo {author} {\bibfnamefont {J.}~\bibnamefont
  {Preskill}}},\Eprint{https://arxiv.org/abs/1203.5813}{arXiv:1203.5813} (\bibinfo {year} {2012})
  \BibitemShut {NoStop}%
  \bibitem{KMR05}
  \BibitemOpen
  \bibfield  {author} {\bibinfo {author} {\bibfnamefont {D.}~\bibnamefont
  {Kenigsberg}}, \bibinfo{author}{\bibfnamefont{T.}~\bibnamefont{Mor}} and \bibinfo{author}{\bibfnamefont{G.}~\bibnamefont{Ratsaby}} },href {\doibase
  http://dx.doi.org/10.1117/12.617175} {\bibfield  {journal} {\bibinfo  {journal}
  {Quant. Inf. Comp.}\ }\textbf {\bibinfo {volume} {6}},\ \bibinfo
  {pages} {, 606} (\bibinfo {year} {2006})}\BibitemShut {NoStop}
\bibitem [{\citenamefont {Vedral}(2010)}]{arXiv:0906.3656}%
  \BibitemOpen
  \bibfield  {author} {\bibinfo {author} {\bibfnamefont {V.}~\bibnamefont
  {Vedral}},\ }\href {\doibase 10.1007/s10701-010-9452-0} {\bibfield  {journal}
  {\bibinfo  {journal} {Found. Phys.}\ }\textbf {\bibinfo {volume} {40}},\
  \bibinfo {pages} {1141} (\bibinfo {year} {2010})}\BibitemShut {NoStop}%
\bibitem [{\citenamefont {Biham}\ \emph {et~al.}(2004)\citenamefont {Biham},
  \citenamefont {Brassard}, \citenamefont {Kenigsberg},\ and\ \citenamefont
  {Mor}}]{arXiv:quant-ph/030618}%
  \BibitemOpen
  \bibfield  {author} {\bibinfo {author} {\bibfnamefont {E.}~\bibnamefont
  {Biham}}, \bibinfo {author} {\bibfnamefont {G.}~\bibnamefont {Brassard}},
  \bibinfo {author} {\bibfnamefont {D.}~\bibnamefont {Kenigsberg}}, \ and\
  \bibinfo {author} {\bibfnamefont {T.}~\bibnamefont {Mor}},\ }\href@noop {}
  {\bibfield  {journal} {\bibinfo  {journal} {Theoretical Computer Science}\
  }\textbf {\bibinfo {volume} {320}},\ \bibinfo {pages} {15} (\bibinfo {year}
  {2004})}\BibitemShut {NoStop}%
\bibitem [{\citenamefont {Kay}(2015)}]{Alastair15}%
  \BibitemOpen
  \bibfield  {author} {\bibinfo {author} {\bibfnamefont {A.}~\bibnamefont
  {Kay}},\ }\href {\doibase 10.1103/PhysRevA.92.062329} {\bibfield  {journal}
  {\bibinfo  {journal} {Phys. Rev. A}\ }\textbf {\bibinfo {volume} {92}},\
  \bibinfo {pages} {062329} (\bibinfo {year} {2015})}\BibitemShut {NoStop}%
  \bibitem [{\citenamefont {Boixo}\ \emph {et~al.}(2014)\citenamefont {Boixo},
  \citenamefont {Ronnow}, \citenamefont {Isakov}, \citenamefont {Wang},
  \citenamefont {Wecker}, \citenamefont {Lidar}, \citenamefont {Martinis},\
  and\ \citenamefont {Troyer}}]{Boixo2014}%
  \BibitemOpen
  \bibfield  {author} {\bibinfo {author} {\bibfnamefont {S.}~\bibnamefont
  {Boixo}}, \bibinfo {author} {\bibfnamefont {T.~F.}\ \bibnamefont {Ronnow}},
  \bibinfo {author} {\bibfnamefont {S.~V.}\ \bibnamefont {Isakov}}, \bibinfo
  {author} {\bibfnamefont {Z.}~\bibnamefont {Wang}}, \bibinfo {author}
  {\bibfnamefont {D.}~\bibnamefont {Wecker}}, \bibinfo {author} {\bibfnamefont
  {D.~A.}\ \bibnamefont {Lidar}}, \bibinfo {author} {\bibfnamefont {J.~M.}\
  \bibnamefont {Martinis}}, \ and\ \bibinfo {author} {\bibfnamefont
  {M.}~\bibnamefont {Troyer}},\ }\href {\doibase 10.1038/nphys2900} {\bibfield
  {journal} {\bibinfo  {journal} {Nature Physics}\ }\textbf {\bibinfo {volume}
  {10}},\ \bibinfo {pages} {218} (\bibinfo {year} {2014})}\BibitemShut
  {NoStop}%
\bibitem [{\citenamefont {Aaronson}\ and\ \citenamefont
  {Arkhipov}(2011)}]{Aaronson2010}%
  \BibitemOpen
  \bibfield  {author} {\bibinfo {author} {\bibfnamefont {S.}~\bibnamefont
  {Aaronson}}\ and\ \bibinfo {author} {\bibfnamefont {A.}~\bibnamefont
  {Arkhipov}},\ }in\ \href {\doibase 10.1145/1993636.1993682} {\emph {\bibinfo
  {booktitle} {STOC '11}}}\ (\bibinfo
  {publisher} {ACM},\ \bibinfo {address} {New York, NY, USA},\ \bibinfo {year}
  {2011})\ pp.\ \bibinfo {pages} {333--342}\BibitemShut {NoStop}%
\bibitem [{\citenamefont {Shepherd}\ and\ \citenamefont
  {Bremner}(2008)}]{Shepherd2008}%
  \BibitemOpen
  \bibfield  {author} {\bibinfo {author} {\bibfnamefont {D.}~\bibnamefont
  {Shepherd}}\ and\ \bibinfo {author} {\bibfnamefont {M.~J.}\ \bibnamefont
  {Bremner}},\ }\href@noop {} {\bibfield  {journal} {\bibinfo  {journal} {Proc.
  R. Soc. A}\ }\textbf {\bibinfo {volume} {465,}},\ \bibinfo {pages} {1413}
  (\bibinfo {year} {2008})}\BibitemShut {NoStop}%
  \bibitem [{\citenamefont {Knill}\ and\ \citenamefont {Laflamme}(1998)}]{KL}%
  \BibitemOpen
  \bibfield  {author} {\bibinfo {author} {\bibfnamefont {E.}~\bibnamefont
  {Knill}}\ and\ \bibinfo {author} {\bibfnamefont {R.}~\bibnamefont
  {Laflamme}},\ }\href {\doibase 10.1103/PhysRevLett.81.5672} {\bibfield
  {journal} {\bibinfo  {journal} {Phys. Rev. Lett.}\ }\textbf {\bibinfo
  {volume} {81}},\ \bibinfo {pages} {5672} (\bibinfo {year} {1998})},\ \Eprint
  {http://arxiv.org/abs/quant-ph/9802037v1} {arXiv:quant-ph/9802037v1}
  \BibitemShut {NoStop}%
\bibitem [{\citenamefont {Datta}\ \emph {et~al.}(2005)\citenamefont {Datta},
  \citenamefont {Flammia},\ and\ \citenamefont {Caves}}]{Datta2005}%
  \BibitemOpen
  \bibfield  {author} {\bibinfo {author} {\bibfnamefont {A.}~\bibnamefont
  {Datta}}, \bibinfo {author} {\bibfnamefont {S.~T.}\ \bibnamefont {Flammia}},
  \ and\ \bibinfo {author} {\bibfnamefont {C.~M.}\ \bibnamefont {Caves}},\
  }\href {\doibase 10.1103/physreva.72.042316} {\bibfield  {journal} {\bibinfo
  {journal} {Phys. Rev. A}\ }\textbf {\bibinfo {volume} {72}},\ \bibinfo
  {pages} {042316} (\bibinfo {year} {2005})}\BibitemShut {NoStop}%
\bibitem [{\citenamefont {Shor}\ and\ \citenamefont {Jordan}(2007)}]{Shor2007}%
  \BibitemOpen
  \bibfield  {author} {\bibinfo {author} {\bibfnamefont {P.~W.}\ \bibnamefont
  {Shor}}\ and\ \bibinfo {author} {\bibfnamefont {S.~P.}\ \bibnamefont
  {Jordan}},\ }\href@noop {} {\bibfield  {journal} {\bibinfo  {journal}
  {Quantum Information and Computation Vol.}\ }\textbf {\bibinfo {volume}
  {8}},\ \bibinfo {pages} {pg.681} (\bibinfo {year} {2007})}\BibitemShut
  {NoStop}%
\bibitem [{\citenamefont {Passante}\ \emph {et~al.}(2009)\citenamefont
  {Passante}, \citenamefont {Moussa}, \citenamefont {Ryan},\ and\ \citenamefont
  {Laflamme}}]{Passante2009}%
  \BibitemOpen
  \bibfield  {author} {\bibinfo {author} {\bibfnamefont {G.}~\bibnamefont
  {Passante}}, \bibinfo {author} {\bibfnamefont {O.}~\bibnamefont {Moussa}},
  \bibinfo {author} {\bibfnamefont {C.~a.}\ \bibnamefont {Ryan}}, \ and\
  \bibinfo {author} {\bibfnamefont {R.}~\bibnamefont {Laflamme}},\ }\href
  {\doibase 10.1103/PhysRevLett.103.250501} {\bibfield  {journal} {\bibinfo
  {journal} {Phys. Rev. Lett.}\ }\textbf {\bibinfo {volume} {103}},\
  \bibinfo {pages} {250501} (\bibinfo {year} {2009})}\BibitemShut {NoStop}%
\bibitem [{\citenamefont {Jozsa}\ and\ \citenamefont
  {Linden}(2003)}]{Jozsa2003}%
  \BibitemOpen
  \bibfield  {author} {\bibinfo {author} {\bibfnamefont {R.}~\bibnamefont
  {Jozsa}}\ and\ \bibinfo {author} {\bibfnamefont {N.}~\bibnamefont {Linden}},\
  }\href {\doibase 10.1098/rspa.2002.1097} {\bibfield  {journal} {\bibinfo
  {journal} {Proceedings of the Royal Society A}\ }\textbf {\bibinfo {volume}
  {459}},\ \bibinfo {pages} {2011} (\bibinfo {year} {2003})}\BibitemShut
  {NoStop}%
\bibitem [{\citenamefont {Vidal}(2003)}]{arXiv:quant-ph/0301063}%
  \BibitemOpen
  \bibfield  {author} {\bibinfo {author} {\bibfnamefont {G.}~\bibnamefont
  {Vidal}},\ }\href {\doibase 10.1103/PhysRevLett.91.147902} {\bibfield
  {journal} {\bibinfo  {journal} {Phys. Rev. Lett.}\ }\textbf {\bibinfo
  {volume} {91}},\ \bibinfo {pages} {147902} (\bibinfo {year}
  {2003})}\BibitemShut {NoStop}%
\bibitem [{\citenamefont {Datta}\ and\ \citenamefont
  {Vidal}(2007)}]{Datta2007}%
  \BibitemOpen
  \bibfield  {author} {\bibinfo {author} {\bibfnamefont {A.}~\bibnamefont
  {Datta}}\ and\ \bibinfo {author} {\bibfnamefont {G.}~\bibnamefont {Vidal}},\
  }\href {\doibase 10.1103/physreva.75.042310} {\bibfield  {journal} {\bibinfo
  {journal} {Phys. Rev. A}\ }\textbf {\bibinfo {volume} {75}} (\bibinfo
  {year} {2007}),\ 10.1103/physreva.75.042310}\BibitemShut {NoStop}%
\bibitem [{\citenamefont {Datta}\ \emph {et~al.}(2008)\citenamefont {Datta},
  \citenamefont {Shaji},\ and\ \citenamefont {Caves}}]{DattaShajiCaves08}%
  \BibitemOpen
  \bibfield  {author} {\bibinfo {author} {\bibfnamefont {A.}~\bibnamefont
  {Datta}}, \bibinfo {author} {\bibfnamefont {A.}~\bibnamefont {Shaji}}, \ and\
  \bibinfo {author} {\bibfnamefont {C.}~\bibnamefont {Caves}},\ }\href
  {\doibase 10.1103/PhysRevLett.100.050502} {\bibfield  {journal} {\bibinfo
  {journal} {Phys. Rev. Lett.}\ }\textbf {\bibinfo {volume} {100}},\ \bibinfo
  {pages} {050502} (\bibinfo {year} {2008})},\ \Eprint
  {http://arxiv.org/abs/0709.0548} {arXiv:0709.0548} \BibitemShut {NoStop}%
\bibitem [{\citenamefont {Datta}(2008)}]{datta}%
  \BibitemOpen
  \bibfield  {author} {\bibinfo {author} {\bibfnamefont {A.}~\bibnamefont
  {Datta}},\ } {Ph.D. thesis},\ \bibinfo
  {school} {The University of New Mexico} (\bibinfo {year} {2008})\BibitemShut
  {NoStop}%
\bibitem [{\citenamefont {Datta}\ and\ \citenamefont
  {Shaji}(2011)}]{arXiv:1109.5549}%
  \BibitemOpen
  \bibfield  {author} {\bibinfo {author} {\bibfnamefont {A.}~\bibnamefont
  {Datta}}\ and\ \bibinfo {author} {\bibfnamefont {A.}~\bibnamefont {Shaji}},\
  }\href@noop {} {\bibfield  {journal} {\bibinfo  {journal} {Int. J. Quant.
  Info.}\ }\textbf {\bibinfo {volume} {9}},\ \bibinfo {pages} {1787} (\bibinfo
  {year} {2011})},\ \Eprint {http://arxiv.org/abs/1109.5549} {arXiv:1109.5549
  [quant-ph]} \BibitemShut {NoStop}%
\bibitem [{\citenamefont {Laflamme}\ \emph {et~al.}(2002)\citenamefont
  {Laflamme}, \citenamefont {Cory}, \citenamefont {Negrevergne},\ and\
  \citenamefont {Viola}}]{arXiv:quant-ph/0110029}%
  \BibitemOpen
  \bibfield  {author} {\bibinfo {author} {\bibfnamefont {R.}~\bibnamefont
  {Laflamme}}, \bibinfo {author} {\bibfnamefont {D.~G.}\ \bibnamefont {Cory}},
  \bibinfo {author} {\bibfnamefont {C.}~\bibnamefont {Negrevergne}}, \ and\
  \bibinfo {author} {\bibfnamefont {L.}~\bibnamefont {Viola}},\ }\href@noop {}
  {\bibfield  {journal} {\bibinfo  {journal} {Quat. Inf. Comp.}\ }\textbf
  {\bibinfo {volume} {2}},\ \bibinfo {pages} {166} (\bibinfo {year}
  {2002})}\BibitemShut {NoStop}%
\bibitem [{\citenamefont {Fanchini}\ \emph {et~al.}(2011)\citenamefont
  {Fanchini}, \citenamefont {Cornelio}, \citenamefont {De~Oliveira},\ and\
  \citenamefont {Caldeira}}]{arXiv:1006.2460}%
  \BibitemOpen
  \bibfield  {author} {\bibinfo {author} {\bibfnamefont {F.~F.}\ \bibnamefont
  {Fanchini}}, \bibinfo {author} {\bibfnamefont {M.~F.}\ \bibnamefont
  {Cornelio}}, \bibinfo {author} {\bibfnamefont {M.~C.}\ \bibnamefont
  {De~Oliveira}}, \ and\ \bibinfo {author} {\bibfnamefont {A.~O.}\ \bibnamefont
  {Caldeira}},\ }\href@noop {} {\bibfield  {journal} {\bibinfo  {journal}
  {Phys. Rev. A}\ }\textbf {\bibinfo {volume} {84}},\ \bibinfo {pages} {12313}
  (\bibinfo {year} {2011})}\BibitemShut {NoStop}%
\bibitem [{\citenamefont {Rieffel}\ and\ \citenamefont
  {Wiseman}(2014)}]{Rieffel2014}%
  \BibitemOpen
  \bibfield  {author} {\bibinfo {author} {\bibfnamefont {E.~G.}\ \bibnamefont
  {Rieffel}}\ and\ \bibinfo {author} {\bibfnamefont {H.~M.}\ \bibnamefont
  {Wiseman}},\ }\href {\doibase 10.1103/PhysRevA.89.032323} {\bibfield
  {journal} {\bibinfo  {journal} {Phys. Rev. A}\ }\textbf {\bibinfo
  {volume} {89}},\ \bibinfo {pages} {032323} (\bibinfo {year}
  {2014})}\BibitemShut {NoStop}%
\bibitem [{\citenamefont {Modi}\ \emph {et~al.}(2012)\citenamefont {Modi},
  \citenamefont {Brodutch}, \citenamefont {Cable}, \citenamefont {Paterek},\
  and\ \citenamefont {Vedral}}]{DiscordRMP}%
  \BibitemOpen
  \bibfield  {author} {\bibinfo {author} {\bibfnamefont {K.}~\bibnamefont
  {Modi}}, \bibinfo {author} {\bibfnamefont {A.}~\bibnamefont {Brodutch}},
  \bibinfo {author} {\bibfnamefont {H.}~\bibnamefont {Cable}}, \bibinfo
  {author} {\bibfnamefont {T.}~\bibnamefont {Paterek}}, \ and\ \bibinfo
  {author} {\bibfnamefont {V.}~\bibnamefont {Vedral}},\ }\href {\doibase
  10.1103/revmodphys.84.1655} {\bibfield  {journal} {\bibinfo  {journal} {Rev.
  Mod. Phys.}\ }\textbf {\bibinfo {volume} {84}},\ \bibinfo {pages} {1655}
  (\bibinfo {year} {2012})}\BibitemShut {NoStop}%
  \bibitem [{\citenamefont {Ollivier}\ and\ \citenamefont
  {Zurek}(2001)}]{OllivierZurek01}%
  \BibitemOpen
  \bibfield  {author} {\bibinfo {author} {\bibfnamefont {H.}~\bibnamefont
  {Ollivier}}\ and\ \bibinfo {author} {\bibfnamefont {W.~H. W. W.~H.}\
  \bibnamefont {Zurek}},\ }\href {\doibase 10.1103/PhysRevLett.88.017901}
  {\bibfield  {journal} {\bibinfo  {journal} {Phys. Rev. Lett.}\ }\textbf
  {\bibinfo {volume} {88}},\ \bibinfo {pages} {17901} (\bibinfo {year}
  {2001})}\BibitemShut {NoStop}%
\bibitem [{\citenamefont {Boyer}(2016)}]{Boyer2016}%
  \BibitemOpen
  \bibfield  {author} {\bibinfo {author} {\bibfnamefont {M.}~\bibnamefont
  {Boyer}},\ }
  \Eprint  {http://arxiv.org/abs/1608.08136} {arXiv:1608.08136 [quant-ph]}  (\bibinfo {year} {2007})\BibitemShut
  {NoStop}%
\bibitem [{\citenamefont {Ferraro}\ \emph {et~al.}(2010)\citenamefont
  {Ferraro}, \citenamefont {Aolita}, \citenamefont {Cavalcanti}, \citenamefont
  {Cucchietti},\ and\ \citenamefont {Ac\'{\i}n}}]{Ferraroetal10}%
  \BibitemOpen
  \bibfield  {author} {\bibinfo {author} {\bibfnamefont {A.}~\bibnamefont
  {Ferraro}}, \bibinfo {author} {\bibfnamefont {L.}~\bibnamefont {Aolita}},
  \bibinfo {author} {\bibfnamefont {D.}~\bibnamefont {Cavalcanti}}, \bibinfo
  {author} {\bibfnamefont {F.~M.}\ \bibnamefont {Cucchietti}}, \ and\ \bibinfo
  {author} {\bibfnamefont {A.}~\bibnamefont {Ac\'{\i}n}},\ }\href {\doibase
  10.1103/PhysRevA.81.052318} {\bibfield  {journal} {\bibinfo  {journal} {Phys.
  Rev. A}\ }\textbf {\bibinfo {volume} {81}},\ \bibinfo {pages} {052318}
  (\bibinfo {year} {2010})}\BibitemShut {NoStop}%
   \bibitem{Neg}
  \BibitemOpen
  \bibfield  {author} {\bibinfo {author} {\bibfnamefont {G.}~\bibnamefont
  {Vidal}} and \bibinfo{author}{\bibfnamefont{R. F.}~\bibnamefont{Werner}} },\href {\doibase
  http://dx.doi.org/doi:10.1103/PhysRevA.65.032314} {\bibfield  {journal} {\bibinfo  {journal}
  {Phys. Rev. A}\ }\textbf {\bibinfo {volume} {65}},\ \bibinfo
  {pages} {032314} (\bibinfo {year} {2002})}\BibitemShut {NoStop}
\bibitem [{\citenamefont {den Nest}(2012)}]{Nest2012}%
  \BibitemOpen
  \bibfield  {author} {\bibinfo {author} {\bibfnamefont {M.~V.}\ \bibnamefont
  {den Nest}},\ }\href@noop {} {\bibfield  {journal} {\bibinfo  {journal}
  {Phys. Rev. Lett.}\ }\textbf {\bibinfo {volume} {110}},\ \bibinfo {pages}
  {60504} (\bibinfo {year} {2012})}\BibitemShut {NoStop}%
\bibitem [{\citenamefont {Boyer}\ and\ \citenamefont {Mor}(2014)}]{TPNC2014}%
  \BibitemOpen
  \bibfield  {author} {\bibinfo {author} {\bibfnamefont {M.}~\bibnamefont
  {Boyer}}\ and\ \bibinfo {author} {\bibfnamefont {T.}~\bibnamefont {Mor}},\
  }\enquote {\bibinfo {title} {{Theory and Practice of Natural Computing:  TPNC 2014}},}\bibinfo {pages}
  {107--118},\ \ (\bibinfo  {publisher} {Springer International
  Publishing},\ \bibinfo {address} {Cham},\ \bibinfo {year} {2014})
  \BibitemShut {NoStop}%
\bibitem [{\citenamefont {Gurvits}\ and\ \citenamefont
  {Barnum}(2002)}]{Gurvits2002}%
  \BibitemOpen
  \bibfield  {author} {\bibinfo {author} {\bibfnamefont {L.}~\bibnamefont
  {Gurvits}}\ and\ \bibinfo {author} {\bibfnamefont {H.}~\bibnamefont
  {Barnum}},\ }\href {\doibase 10.1103/PhysRevA.66.062311} {\bibfield
  {journal} {\bibinfo  {journal} {Phys. Rev. A}\ }\textbf {\bibinfo
  {volume} {66}},\ \bibinfo {pages} {062311} (\bibinfo {year}
  {2002})}\BibitemShut {NoStop}%
\bibitem [{\citenamefont {Braunstein}\ \emph {et~al.}(1999)\citenamefont
  {Braunstein}, \citenamefont {Caves}, \citenamefont {Jozsa}, \citenamefont
  {Linden}, \citenamefont {Popescu},\ and\ \citenamefont
  {Schack}}]{PhysRevLett.83.1054}%
  \BibitemOpen
  \bibfield  {author} {\bibinfo {author} {\bibfnamefont {S.~L.}\ \bibnamefont
  {Braunstein}}, \bibinfo {author} {\bibfnamefont {C.~M.}\ \bibnamefont
  {Caves}}, \bibinfo {author} {\bibfnamefont {R.}~\bibnamefont {Jozsa}},
  \bibinfo {author} {\bibfnamefont {N.}~\bibnamefont {Linden}}, \bibinfo
  {author} {\bibfnamefont {S.}~\bibnamefont {Popescu}}, \ and\ \bibinfo
  {author} {\bibfnamefont {R.}~\bibnamefont {Schack}},\ }\href {\doibase
  10.1103/PhysRevLett.83.1054} {\bibfield  {journal} {\bibinfo  {journal}
  {Phys. Rev. Lett.}\ }\textbf {\bibinfo {volume} {83}},\ \bibinfo {pages}
  {1054} (\bibinfo {year} {1999})},\ \Eprint
  {http://arxiv.org/abs/quant-ph/9811018} {arXiv:quant-ph/9811018} \BibitemShut
  {NoStop}%
\bibitem [{\citenamefont {Hildebrand}(2007)}]{Hildebrand2007}%
  \BibitemOpen
  \bibfield  {author} {\bibinfo {author} {\bibfnamefont {R.}~\bibnamefont
  {Hildebrand}},\ }\href {\doibase 10.1103/physreva.76.052325} {\bibfield
  {journal} {\bibinfo  {journal} {Phys. Rev. A}\ }\textbf {\bibinfo {volume}
  {76}} (\bibinfo {year} {2007})}\BibitemShut
  {NoStop}%
  \bibitem [{\citenamefont {Johnston}(2013)}]{Johnston2013}%
  \BibitemOpen
  \bibfield  {author} {\bibinfo {author} {\bibfnamefont {N.}~\bibnamefont
  {Johnston}},\ }\href {\doibase 10.1103/physreva.88.062330} {\bibfield
  {journal} {\bibinfo  {journal} {Phys. Rev. A}\ }\textbf {\bibinfo {volume}
  {88}} (\bibinfo {year} {2013})}\BibitemShut
  {NoStop}%
\bibitem [{\citenamefont {Arunachalam}\ \emph {et~al.}(2014)\citenamefont
  {Arunachalam}, \citenamefont {Johnston},\ and\ \citenamefont
  {Russo}}]{Arunachalam2014}%
  \BibitemOpen
  \bibfield  {author} {\bibinfo {author} {\bibfnamefont {S.}~\bibnamefont
  {Arunachalam}}, \bibinfo {author} {\bibfnamefont {N.}~\bibnamefont
  {Johnston}}, \ and\ \bibinfo {author} {\bibfnamefont {V.}~\bibnamefont
  {Russo}},\ }\href@noop {} {\bibfield  {journal} {\bibinfo  {journal} {ArXiv
  e-prints}\ } (\bibinfo {year} {2014})},\ \Eprint
  {http://arxiv.org/abs/1405.5853} {arXiv:1405.5853 [quant-ph]} \BibitemShut
  {NoStop}%
\bibitem [{\citenamefont {Dakić}\ \emph {et~al.}(2010)\citenamefont {Dakić},
  \citenamefont {Vedral},\ and\ \citenamefont {Časlav
  Brukner}}]{DakicVedralBrukner10}%
  \BibitemOpen
  \bibfield  {author} {\bibinfo {author} {\bibfnamefont {B.}~\bibnamefont
  {Dakić}}, \bibinfo {author} {\bibfnamefont {V.}~\bibnamefont {Vedral}}, \
  and\ \bibinfo {author} {\bibnamefont {Časlav Brukner}},\ }\href {\doibase
  10.1103/PhysRevLett.105.190502} {\bibfield  {journal} {\bibinfo  {journal}
  {Phys. Rev. Lett.}\ }\textbf {\bibinfo {volume} {105}},\ \bibinfo
  {pages} {190502} (\bibinfo {year} {2010})}\BibitemShut {NoStop}%
   \bibitem [{\citenamefont {Boyer}\ \emph {et~al.}({\natexlab{a}})\citenamefont
  {Boyer}, \citenamefont {Brodutch},\ and\ \citenamefont {Mor}}]{BBMinprep}%
  \BibitemOpen
  \bibfield  {author} {\bibinfo {author} {\bibfnamefont {M.}~\bibnamefont
  {Boyer}}, \bibinfo {author} {\bibfnamefont {A.}~\bibnamefont {Brodutch}}, \
  and\ \bibinfo {author} {\bibfnamefont {T.}~\bibnamefont {Mor}},\ }\href@noop
  {} {} ({\natexlab{a}}),\ \bibinfo {note} {in preparation}\BibitemShut
  {NoStop}%
\bibitem [{\citenamefont {Groisman}\ \emph {et~al.}(2007)\citenamefont
  {Groisman}, \citenamefont {Kenigsberg},\ and\ \citenamefont
  {Mor}}]{Groisman2007}%
  \BibitemOpen
  \bibfield  {author} {\bibinfo {author} {\bibfnamefont {B.}~\bibnamefont
  {Groisman}}, \bibinfo {author} {\bibfnamefont {D.}~\bibnamefont
  {Kenigsberg}}, \ and\ \bibinfo {author} {\bibfnamefont {T.}~\bibnamefont
  {Mor}},\ }\href@noop {}  (\bibinfo {year} {2007}),\ \Eprint
  {http://arxiv.org/abs/arXiv:quant-ph/0703103} {arXiv:quant-ph/0703103}
  \BibitemShut {NoStop}%
 \end{thebibliography}
\end{document}